\documentclass[11pt]{article}
\usepackage{amsmath,amssymb, fullpage}
\usepackage{algorithmic}
\usepackage{algorithm}

\title{New Results on Quantum Property Testing}
\author{
Sourav Chakraborty\thanks{Centrum Wiskunde \&\ Informatica, Amsterdam. Email: {\tt \{sourav,ariem,rdewolf\}@cwi.nl}. RdW is partially supported by a Vidi grant from the Netherlands Organization for Scientific Research (NWO), and by the European Commission under the Integrated Project Qubit Applications (QAP) funded by the IST directorate as Contract Number 015848.}
\and
Eldar Fischer\thanks{Computer Science Faculty, Israel Institute of Technology (Technion). Email: {\tt eldar@cs.technion.ac.il}. Partially supported by an ERC-2007-StG grant number 202405-2 and by an ISF grant number 1101/06.}
\and
Arie Matsliah$^*$
\and
Ronald de Wolf$^*$
}

\date{}

\newtheorem{theorem}{Theorem}[section]

\newtheorem{lemma}[theorem]{Lemma}
\newtheorem{corollary}[theorem]{Corollary}

\newtheorem{definition}[theorem]{Definition}

\newcommand{\qed}{\rule{2mm}{2mm}}
\newenvironment{proof}{\par\noindent{\bf Proof.}\quad}{\hfill  $\qed$}
\newcommand{\QEstimate}{\mathrm{QEstimate}}

\newcommand{\zo}{{\{0,1\}}}

\newcommand{\poly}{{\mathrm{poly} }}

\newcommand{\PP}{{\cal P}}
\newcommand{\tPP}{{\widetilde{\cal P}}}
\newcommand{\tW}{{\widetilde{W}}}

\newcommand{\ignore}[1]{}
\newcommand{\ket}[1]{| #1 \rangle}
\newcommand{\Exp}{\mathbb{E}}
\newcommand{\Var}{\mathrm{Var}}

\newcommand{\norm}[1]{\lVert{#1}\rVert}

\begin{document}
\maketitle

\begin{abstract}
We present several new examples of speed-ups obtainable by quantum algorithms in the context of property testing.

First, motivated by sampling algorithms, we consider probability distributions given in the form of an oracle $f:[n]\to[m]$.
Here the probability $\PP_f(j)$ of an outcome $j\in[m]$ is the fraction of its domain that $f$ maps to $j$.
We give quantum algorithms for testing whether two such distributions are identical or $\epsilon$-far in
$L_1$-norm. Recently, Bravyi, Hassidim, and Harrow~\cite{BHH10} showed that if
$\PP_f$ and $\PP_g$ are both unknown (i.e., given by oracles $f$ and
$g$), then this testing can be done in roughly $\sqrt{m}$ quantum
queries to the functions. We consider the case where the second distribution is
known, and show that testing can be done with roughly $m^{1/3}$
quantum queries, which we prove to be essentially optimal. In contrast, it is
known that classical testing algorithms need about $m^{2/3}$ queries
in the unknown-unknown case and about $\sqrt{m}$ queries in the
known-unknown case. Based on this result, we also reduce the query
complexity of graph isomorphism testers with quantum oracle access.

While those examples provide polynomial quantum speed-ups, our third example gives a much larger 
improvement (constant quantum queries vs polynomial classical queries)
for the problem of testing periodicity, based on Shor's algorithm and a modification of
a classical lower bound by Lachish and Newman~\cite{lachish&newman:periodicity}.
This provides an alternative to a recent constant-vs-polynomial speed-up due to Aaronson~\cite{aaronson:bqpph}.
\end{abstract}

\newpage

\tableofcontents

\newpage

\section{Introduction}

Since the early 1990s, a number of \emph{quantum} algorithms have been
discovered that have much better query complexity than their best classical
counterparts~\cite{deutsch&jozsa,simon:power,grover:search,ambainis:edj,fgg:nandtreej,acrsz:andor}.
Around the same time, the area of \emph{property testing} gained prominence~\cite{blr:selftest,ggr,fis_sur,ron_sur}.
Here the aim is to design algorithms that can efficiently test whether a given very large piece of data satisfies
some specific property, or is ``far'' from having that property.

Buhrman et al.~\cite{bfnr:qprop} combined these two strands, exhibiting various testing problems
where quantum testers are much more efficient than classical testers.
There has been some recent subsequent work on quantum property testing,
such as the work of Friedl et al.~\cite{fmsp:testhidden} on testing hidden group properties,
Atici and Servedio~\cite{atici&servedio:testing} on testing juntas,
Inui and Le Gall~\cite{inui&legall:testing} on testing group solvability,
Childs and Liu~\cite{childs&liu:testing} on testing bipartiteness and expansion,
Aaronson~\cite{aaronson:bqpph} on ``Fourier checking'',
and Bravyi, Hassidim, and Harrow~\cite{BHH10} on testing distributions.
We will say more about the latter papers below.

In this paper we continue this line of research, coming up with a number of new examples where quantum testers
substantially improve upon their classical counterparts. It should be noted that we do not invent new quantum algorithms
here---rather, we use known quantum algorithms as subroutines in otherwise {\em classical} testing algorithms.

\subsection{Distribution Testing}

How many samples are needed to determine whether two distributions
are identical or have $L_1$-distance more than $\epsilon$? This is a
fundamental problem in statistical hypothesis testing and also
arises in other subjects like property testing and machine learning.

We use the notation $[n] = \{1,2,3,\ldots,n\}$.
For a function $f:[n] \to [m]$, we denote by $\PP_f$ the distribution over $[m]$ in which the weight $\PP_f(j)$ of every $j \in [m]$
is proportional to the number of elements $i \in [n]$ that are mapped to $j$.
We use this form of representation for distributions in order to allow \emph{queries}.
Namely, we assume that the function $f:[n]\to [m]$ is accessible by an oracle of the form
$\ket{x}\ket{b} \mapsto \ket{x}\ket{b \oplus f(x)}$, where $x$ is a $\log n$-bit string,
$b$ and $f(x)$ are $\log m$-bit strings and $\oplus$ is bitwise
addition modulo two. Note that a classical random sample according
to a distribution $\PP_f$ can be simply obtained by picking $i
\in [n]$ uniformly at random and evaluating $f(i)$. In fact, a classical algorithm cannot make a better use of the oracle,
since the actual labels of the domain $[n]$ are irrelevant.
See Section~\ref{sec:sampling} for more on the relation between sampling a distribution and querying a function.

We say that the distribution $\PP_f$ is \emph{known} (or
\emph{explicit}) if the function $f$ is given explicitly, and hence
all probabilities $\PP_f(j)$ can be computed. $\PP_f$ is
\emph{unknown}  (or \emph{black-box}) if we only have oracle access
to the function $f$, and no additional information about $f$ is
given. Two distributions $\PP_f,\PP_g$ defined by functions $f,g:[n]
\to [m]$ are {\em $\epsilon$-far} if the $L_1$-distance between them
is at least $\epsilon$, i.e., $\norm{\PP_f - \PP_g}_1 = \sum_{j=1}^m
|\PP_f(j)-\PP_g(j)| \geq \epsilon$. Note that $f=g$ implies
$\PP_f=\PP_g$ but not vice versa (for instance, permuting $f$ leaves
$\PP_f$ invariant). Two problems of testing distributions can be formally stated as follows:
\begin{itemize}
\item \textbf{unknown-unknown case.}
Given $n,m, \epsilon$ and oracle access to $f,g:[n] \to [m]$, how many queries to $f$ and $g$ are required
in order to determine whether the distributions $\PP_f$ and $\PP_g$ are identical or $\epsilon$-far?

\item \textbf{known-unknown case.}
Given $n,m, \epsilon$, oracle access to $f:[n] \to [m]$ and a known distribution $\PP_g$ (defined by an explicitly given function $g:[n]\to [m]$),
how many queries to $f$ are required to determine whether $\PP_f$ and $\PP_g$ are identical or $\epsilon$-far?
\end{itemize}

If only \emph{classical} queries are allowed (where querying the distribution means asking for a random sample), the answers to these problems are well known. For the unknown-unknown case
 Batu, Fortnow, Rubinfeld, Smith, and White \cite{BFR+00} proved an upper bound of $\widetilde{O}(m^{2/3})$
on the query complexity, and Valiant \cite{Val} proved a matching (up to polylogarithmic factors) lower bound.
For the known-unknown case, Goldreich and Ron \cite{GR} showed a lower bound of $\Omega(\sqrt{m})$ queries and
 Batu, Fischer, Fortnow, Rubinfeld, Smith, and White \cite{BFF+01} proved a nearly
tight upper bound of $\widetilde{O}(\sqrt{m})$ queries.\footnote{These classical lower bounds are stated
in terms of number of samples rather than number of queries, but it is not hard to see that they hold in both models.
In fact, the $\sqrt{m}$ classical query lower bound for the known-unknown case follows by the same argument
as the quantum lower bound in Section~\ref{sec:qlbs}.}

\subsubsection{Testing with Quantum Queries}

Allowing quantum queries for accessing distributions, Bravyi,
Hassidim, and Harrow~\cite{BHH10} recently showed that the
$L_1$-distance between two unknown distributions can actually be
estimated up to small error with only $O(\sqrt{m})$ queries. Their
result implies an $O(\sqrt{m})$ upper bound on the quantum query
complexity for the unknown-unknown testing problem defined above.
In this paper we consider the known-unknown case, and prove nearly tight bounds on its quantum query complexity.

\begin{theorem}\label{thm:1unknown} Given $n,m, \epsilon$, oracle access to $f:[n] \to [m]$ and a known distribution
$\PP_g$ (defined by an explicitly given function $g:[n]\to [m]$), the quantum query complexity of determining whether
$\PP_f$ and $\PP_g$ are identical or $\epsilon$-far is $O(\frac{m^{1/3} \log^2 m\log \log m}{\epsilon^5}) = m^{1/3} \cdot \poly(\frac{1}{\epsilon}, \log m)$.
\end{theorem}

We prove Theorem~\ref{thm:1unknown} in two parts. First, in
Section~\ref{sec:uniform}, we prove that with
$O(\frac{m^{1/3}}{\epsilon^2})$ quantum queries it is possible to test
whether a black-box distribution $\PP_f$ (defined by some $f:[n] \to
[m]$) is $\epsilon$-close to uniform. We actually prove that this can be even
done {\em tolerantly} in a sense, meaning that a distribution that
is close to uniform in the $L_\infty$ norm is accepted
with high probability (see Theorem \ref{theo:unitester} for the formal statement).
Then, in Section~\ref{sec:general}, we use the bucketing
technique (see Section \ref{sec:bucketing}) to reduce the task of testing
closeness to a known distribution to testing uniformity.

We stress that the main difference between the classical algorithm of \cite{BFF+01} and ours is that in \cite{BFF+01} they
check the ``uniformity'' of the unknown distribution in every bucket by approximating the corresponding $L_2$ norms of the conditional
distributions. It is not clear if one can gain anything (in the quantum case) using the same strategy, since
we are not aware of any quantum procedure that can approximate the $L_2$ norm of a distribution with less than $\sqrt{m}$ queries. Hence,
we reduce the main problem {\em directly} to the problem of testing uniformity.
For this reduction to work, the uniformity tester has to be tolerant in the sense mentioned above (see Section~\ref{sec:general} for details).

A different quantum uniformity tester was recently discovered (independently) in \cite{BHH10}.
We note that our version has the advantages of being tolerant, which is crucial for the application above, and it has only polynomial dependence on
$\epsilon$ (instead of exponential), which is essentially optimal.

\subsubsection{Quantum Lower Bounds}
Known quantum query lower bounds for the collision
problem~\cite{aaronson&shi:collision,ambainis:colsmallrange,kutin:collision}
imply that in both known-unknown and unknown-unknown cases roughly $m^{1/3}$ quantum queries
are required. In fact, the lower bound applies even for testing uniformity (see proof in Section~\ref{sec:qlbs}):

\begin{theorem}\label{thm:lb} Given $n,m, \epsilon$ and oracle access to $f:[n] \to [m]$, the quantum query complexity of determining whether
$\PP_f$ is uniform or $\epsilon$-far from uniform is $\Omega(m^{1/3})$.
\end{theorem}

The main remaining open problem is to tighten the bounds on the quantum query complexity for the unknown-unknown case.
It would be very interesting if this case could also be tested using roughly $m^{1/3}$ quantum queries.
In Section~\ref{sec:reconstructionlower} we show that the easiest way to do this (just reconstructing both unknown distributions up to small error) will not work---it requires $\Omega(m/\log m)$ quantum queries.

\subsection{Graph Isomorphism Testing}

Fischer and Matsliah~\cite{FM} studied the problem of testing
graph isomorphism in the dense-graph model, where the graphs are represented by their
adjacency matrices, and querying the graph corresponds to reading a single entry from its adjacency  matrix.
The goal in isomorphism testing is to determine, with high
probability, whether two graphs $G$ and $H$ are isomorphic or
$\epsilon$-far from being isomorphic, making as few queries as possible.
(The graphs are $\epsilon$-far from being isomorphic if at least an $\epsilon$-fraction
of the entries in their adjacency matrices need to be modified in order to make
them isomorphic.)

In \cite{FM} two models were considered:
\begin{itemize}
\item \textbf{unknown-unknown case.} Both $G$ and $H$ are
unknown, and they can only be accessed by querying their adjacency matrices.

\item \textbf{known-unknown case.} The graph $H$ is known (given in advance to the tester), and
the graph $G$ is unknown (can only be accessed by querying its adjacency matrix).
\end{itemize}

As usual, in both models the query complexity is the worst-case number of
queries needed to test whether the graphs are isomorphic. \cite{FM}
give nearly tight bounds of
$\widetilde{\Theta}(\sqrt{|V|})$ on the (classical) query complexity in the
known-unknown model. For the unknown-unknown model they prove an upper
bound of $\widetilde{O}(|V|^{5/4})$ and a lower bound of
$\Omega(|V|)$ on the query complexity.

Allowing quantum queries\footnote{A quantum query to the adjacency
matrix of a graph $G$ can be of the form
$\ket{i,j}\ket{b} \mapsto \ket{i,j}\ket{b\oplus G(i,j)}$, where
$G(i,j)$ is the $(i,j)$-th entry of the adjacency matrix of $G$ and
$\oplus$ is addition modulo two.}, we can use our aforementioned results to prove the following query-complexity bounds
for testing graph isomorphism (see proof in Section \ref{sec:uu_gi}):
\begin{theorem} \label{theo:uu_gi}
The quantum query complexity of testing graph isomorphism in the known-unknown case is
$\widetilde{\Theta}(|V|^{1/3})$, and in the unknown-unknown case it is between $\Omega(|V|^{1/3})$ and $\widetilde{\Theta}(|V|^{7/6})$.
\end{theorem}

\subsection{Periodicity Testing}
The quantum testers mentioned above obtain polynomial speed-ups over their classical counterparts,
and that is the best one can hope to obtain for these problems.
The paper by Buhrman et al.~\cite{bfnr:qprop}, which first studied quantum property testing,
actually provides two super-polynomial separations between quantum and classical testers:
a constant-vs-$\log n$ separation based on the Bernstein-Vazirani algorithm,
and a (roughly) $\log n$-vs-$\sqrt{n}$ separation based on Simon's algorithm.
They posed as an open problem whether there exists a constant-vs-$n$ separation.
Recently, in an attempt to construct oracles to separate BQP from the Polynomial Hierarchy,
Aaronson~\cite{aaronson:bqpph} analyzed the problem of ``Fourier checking'': 
roughly, the input consists of two $m$-bit Boolean functions $f$ and $g$,
such that $g$ is either strongly or weakly correlated with the Fourier transform of $f$ 
(i.e., $g(x)=\mbox{sign}(\hat{f}(x))$ either for most $x$ or for roughly half of the $x$).
He proved that quantum algorithms can decide this with $O(1)$ queries while classical
algorithms need $\Omega(2^{m/4})$ queries.
Viewed as a testing problem on an input of length $n=2\cdot 2^m$ bits, 
this is the first constant-vs-polynomial separation between quantum and classical testers.

In Section~\ref{sec:periodicity} we obtain another separation that is (roughly) constant-vs-$n^{1/4}$.
Our testing problem is reverse-engineered from the periodicity problem solved by
Shor's famous factoring algorithm~\cite{shor:factoring}.
Suppose we are given a function $f:[n]\to [m]$, which we can query in the usual way.
We call $f$ \emph{1-1-$p$-periodic} if the function is injective on $[p]$ and repeats afterwards.
Equivalently:
\begin{quote}
$f(i)=f(j)$ iff $i=j$ mod $p$.
\end{quote}
Note that we need $m\geq p$ to make this possible. In fact, for simplicity we will assume $m\geq n$.
Let $\PP_p$ be the set of functions $f:[n]\to [m]$ that are 1-1-$p$-periodic, and $\PP_{q,r}=\cup_{p=q}^{r}\PP_p$.
The {\sc 1-1-periodicity testing} problem,
with parameters $q\leq r$ and small fixed constant $\epsilon$, is as follows:
\begin{quote}
given an $f$ which is either in $\PP_{q,r}$ or $\epsilon$-far from $\PP_{q,r}$, find out which is the case.
\end{quote}
Note that \emph{for a given $p$} it is easy to test whether $f$ is $p$-periodic or $\epsilon$-far from it:
choose an $i\in[p]$ uniformly at random, and test whether $f(i)=f(i+kp)$ for a random positive integer $k$.
If $f$ is $p$-periodic then these values will be the same, but if $f$ is $\epsilon$-far from $p$-periodic
then we will detect this with constant probability.
However, $r-q+1$ different values of $p$ are possible in $\PP_{q,r}$, and
we will see below that we cannot efficiently test all of them---at least not in the classical case.
In the \emph{quantum} case, however, we can.

\begin{theorem}\label{thperiodicity}
There is a quantum tester for $\PP_{\sqrt{n}/4,\sqrt{n}/2}$ using $O(1)$ queries (and polylog$(n)$ time),
while for every even integer $r\in[2,n/2)$, every classical tester for $\PP_{r/2,r}$ needs to make $\Omega(\sqrt{r/\log r\log n})$ queries. In particular, testing $\PP_{\sqrt{n}/4,\sqrt{n}/2}$ requires $\Omega(n^{1/4}/\log n)$ classical queries.
\end{theorem}

The quantum upper bound is obtained by a small modification of Shor's algorithm: use Shor to find the period (if there is one)
and then test this purported period with another $O(1)$ queries.
The classical lower is based on ideas from Lachish and Newman~\cite{lachish&newman:periodicity},
who proved classical testing lower bounds for more general periodicity-testing problems.
However, while we follow their general outline, we need to modify their proof
since it specifically applies to functions with range $\zo$, which is different from our 1-1 case.
The requirement of being 1-1 within each period is crucial for the upper bound---quantum 
algorithms need about $\sqrt{n}$ queries to find the period of functions with range $\zo$.
While our separation is slightly weaker than Aaronson's separation for Fourier checking 
(our classical lower bound is $n^{1/4}/\log n$ instead $n^{1/4}$), the problem of periodicity testing 
is arguably more natural, and it may have more applications than Fourier checking.

\section{Preliminaries}
For any distribution $\PP$ on $[m]$ we denote by $\PP(j)$ the
probability mass of $j \in [m]$ and for any $M \subseteq [m]$ we denote by $\PP(M)$ the sum $\sum_{j \in M} \PP(j)$.
For a function $f:[n] \to [m]$, we denote by $\PP_f$ the distribution over $[m]$ in which the weight $\PP_f(j)$ of every $j \in [m]$
is proportional to the number of elements $i \in [n]$ that are mapped to $j$.
Formally, for all $j \in [m]$ we define $\PP_f(j) \triangleq \Pr_{i \sim U}[f(i) = j] = \frac{|f^{-1}(j)|}{n}$,
where $U$ is the uniform distribution on $[n]$, that is $U(i) = 1/n$
for all $i\in[n]$. Whenever the domain is clear from context (and
may be something other than $[n]$), we also use $U$ to denote the
uniform distribution on that domain.

Let $\norm{\cdot}_1$ and $\norm{\cdot}_\infty$ stand for $L_1$-norm
and $L_\infty$-norm respectively. Two distributions $\PP_f,\PP_g$
defined by functions $f,g:[n] \to [m]$ are {\em $\epsilon$-far} if the
$L_1$-distance between them is at least $\epsilon$. Namely, $\PP_f$
is $\epsilon$-far from $\PP_g$ if $\norm{\PP_f - \PP_g}_1 =
\sum_{j=1}^m |\PP_f(j)-\PP_g(j)| \geq \epsilon$.

\subsection{Bucketing} \label{sec:bucketing}

Bucketing is a general tool, introduced in \cite{BFR+00, BFF+01},
that decomposes any explicitly given distribution into a collection
of distributions that are almost uniform. In this section we recall
the bucketing technique and the lemmas (from \cite{BFR+00,BFF+01})
that we will need for our proofs.

\begin{definition} Given a distribution $\PP$ over $[m]$, and
$M \subseteq [m]$ such that $\PP(M) > 0$, the \emph{restriction}
$\PP_{|M}$ is a distribution over $M$ with $\PP_{|M}(i) =
\PP(i)/\PP(M)$.

Given a partition $\mathcal{M} = \{M_0, M_1, \dots, M_k\}$ of $[m]$, we denote by $\PP_{\langle \mathcal{M}\rangle}$ the distribution
 over $\{0\}\cup[k]$ in which $\PP_{\langle \mathcal{M}\rangle}(i) = \PP(M_i)$.
\end{definition}

Given an explicit distribution $\PP$ over $[m]$, $Bucket(\PP, [m],
\epsilon)$ is a procedure that generates a partition $\{M_0, M_1, \dots, M_k\}$ of the domain $[m]$, where $k
= \frac{2\log m}{\log(1 + \epsilon)}$. This partition satisfies the following conditions:
\begin{itemize}
 \item $M_0 = \{j \in [m] \mid \PP(j) < \frac{1}{m\log m}\}$;
 \item for all $i \in [k]$, $M_i = \left\{ j \in [m] \mid \frac{(1+\epsilon)^{i-1}}{m\log m} \leq \PP(j)
< \frac{(1+\epsilon)^{i}}{m\log m}\right\}$.
\end{itemize}

\begin{lemma}[\cite{BFF+01}]\label{lem:bucket} Let $\PP$ be a distribution over $[m]$ and let
$\{M_0, M_1, \dots, M_k\} \gets Bucket(\PP, [m], \epsilon)$. Then $(i)$ $\PP(M_0) \leq 1/\log m$; $(ii)$
 for all $i\in [k]$, $\norm{\PP_{|M_i} - U_{|M_i}}_1 \leq \epsilon$.
\end{lemma}

\begin{lemma}[\cite{BFF+01}]\label{lem:partition} Let $\PP, \PP'$ be two distributions over $[m]$ and let
$\mathcal{M} = \{M_0, M_1, \dots, M_k\}$ be a partition of $[m]$. If  $\norm{\PP_{|M_i} - \PP'_{|M_i}}_1 \leq \epsilon_1$ for every $i \in [k]$
and if in addition $\norm{\PP_{\langle \mathcal{M} \rangle} - \PP'_{\langle \mathcal{M}
\rangle}}_1 \leq \epsilon_2$, then $\norm{\PP - \PP'}_1 \leq \epsilon_1 + \epsilon_2$.
\end{lemma}

\begin{corollary}\label{coro:partition} Let $\PP, \PP'$ be two distributions over $[m]$ and let
$\mathcal{M} = \{M_0, M_1, \dots, M_k\}$ be a partition of $[m]$. If $\norm{\PP_{|M_i} - \PP'_{|M_i}}_1 \leq \epsilon_1$
for every $i \in [k]$ such that $\PP(M_i) \geq \epsilon_3/k$,
and if in addition $\norm{\PP_{\langle \mathcal{M} \rangle} - \PP'_{\langle \mathcal{M}
\rangle}}_1 \leq \epsilon_2$, then $\norm{\PP - \PP'}_1 \leq 2(\epsilon_1 + \epsilon_2 + \epsilon_3)$.
\end{corollary}

\subsection{Quantum Queries and Approximate Counting}
Since we only use specific quantum procedures as a black-box in otherwise classical algorithms,
we will not explain the model of quantum query algorithms in much detail (see~\cite{nielsen&chuang:qc,buhrman&wolf:dectreesurvey} for that).
Suffice it to say that the function $f$ is assumed to be accessible by the oracle unitary
transformation $O_f$, which acts on a $(\log n+ \log m)$-qubit
space by sending the basis vector $\ket{x} \ket{b}$ to $\ket{x}\ket{b \oplus f(x)}$ where $\oplus$ is bitwise addition modulo two.

For any set $S\subseteq [m]$, let $U^S_f$ denote the unitary
transformation which maps $\ket{x}\ket{b}$ to $\ket{x}\ket{b\oplus
1}$ if $f(x) \in S$, and to $\ket{x}\ket{b\oplus 0}$ otherwise.
This unitary transformation can be easily
implemented using $\log m$ ancilla bits and two queries to $O_f$.\footnote{We need \emph{two}
queries to $f$ instead of one, because the quantum algorithm has to ``uncompute'' the first query in order to clean up its workspace.}
If $f_S:[n]\to \{0,1\}$ is defined as $f_S(x) = 1$ if and only if $f(x)
\in S$, then the unitary transformation $U^S_f$ acts as an oracle to
the function $f_S$. Brassard, H{\o}yer, Mosca, and Tapp~\cite[Theorem~13]{bhmt:countingj}
gave an algorithm to approximately count the size of certain sets.

\begin{theorem}[\cite{bhmt:countingj}]\label{thm:quantum}
For every positive integer $q$ and $\ell > 1$, and given quantum
oracle access to a Boolean function $h:[n] \to \{0,1\}$, there is an
algorithm that makes $q$ queries to $h$ and outputs an estimate $t'$
to $t = |h^{-1}(1)|$ such that
$|t' - t| \leq 2\pi \ell \frac{\sqrt{t(n - t)}}{q} + \pi^2\ell^2 \frac{n}{q^2}$
with probability at least $1-1/2(\ell-1)$.
\end{theorem}

The following lemma allows us to estimate the size of the pre-image of a set $S \subseteq [m]$ under $f$. It follows easily from 
Theorem \ref{thm:quantum}.

\begin{lemma}\label{lem:quantum}
For every $\delta \in [0,1]$, for every oracle $O_f$ for the function $f:[n] \to [m]$, and for every set
$S\subseteq [m]$, there is a quantum algorithm
\textbf{$\QEstimate(f, S, \delta)$} that makes
$O(m^{1/3}/\delta)$
queries to $f$ and, with probability at least $5/6$,
outputs an estimate $p'$ to $p = \PP_f(S) =|f^{-1}(S)|/n$ such that $|p'-p|\leq\frac{\delta \sqrt{p}}{m^{1/3}}+\frac{\delta^2}{m^{2/3}}$.
\end{lemma}

\begin{proof}
The algorithm is basically required to estimate $|f_S^{-1}(1)|$.
Using two queries to the oracle $O_f$ we can construct a unitary $U^S_f$ that acts like an
oracle for the Boolean function $f_S$.  Estimate $t=|f_S^{-1}(1)|$ using the algorithm in
Theorem~\ref{thm:quantum}, with $q=cm^{1/3}/\delta$ queries.
Choosing $c$ a sufficiently large constant, with probability at least $5/6$,
the estimate $t'$  satisfies $|t-t'| \leq \frac{\delta \sqrt{t(n-t)}}{m^{1/3}}+\frac{\delta^2n}{m^{2/3}}$. Setting $p' = t'/n$ and bounding $(n-t)$ with $n$  we get that
with probability at least $5/6$,
$|p-p'| = \frac{|t-t'|}{n} \leq \frac{\delta \sqrt{p}}{m^{1/3}}+\frac{\delta^2}{m^{2/3}}$.
\end{proof}

\section{Proof of Theorem~\ref{thm:1unknown}}

\subsection{Testing Uniformity Tolerantly}\label{sec:uniform}
Given $\epsilon > 0$
and oracle access to a function $f:[n] \to [m]$, our task is to
distinguish the case $\norm{\PP_f - U}_1 \geq \epsilon$ from the
case $\norm{\PP_f - U}_\infty \leq \epsilon/4m$. Note that this is
a stronger condition than the one required for the usual testing
task, where the goal is to distinguish the case $\norm{\PP_f - U}_1
\geq \epsilon$ from $\norm{\PP_f - U}_\infty = \norm{\PP_f - U}_1 = 0$.

\begin{theorem}\label{theo:unitester} There is a quantum testing
algorithm (Algorithm \ref{alg:uni}, below) that given $\epsilon>0$ and oracle access to a function $f:[n]\to [m]$ makes
$O(\frac{m^{1/3}}{\epsilon^2})$ quantum queries and with probability at least $2/3$ outputs
 REJECT if $\norm{\PP_f- U}_1 \geq \epsilon$, and ACCEPT if $\norm{\PP_f - U}_\infty \leq \epsilon/4m$.
\end{theorem}
\begin{algorithm}
\caption{(Tests closeness to the uniform distribution.)}
\label{alg:uni}
\begin{algorithmic}[1a]
\STATE pick a set $T \subseteq [n]$ of $t = m^{1/3}$ indices uniformly at random
\STATE query $f$ on all indices in $T$; set $S \gets \{f(i)\mid i \in T\}$
\IF {$f(i)=f(j)$ for some $i,j \in T$, $i \neq j$ (or equivalently, $|S| < t$)}
        \STATE REJECT
\ENDIF \STATE $p' \gets \QEstimate(f, S, \delta)$, with $\delta \triangleq \frac{\epsilon^2}{320}$
\IF {$|p'-\frac{t}{m}|\leq 32\delta\frac{t}{m}$}
    \STATE ACCEPT
\ELSE
    \STATE REJECT
\ENDIF
\end{algorithmic}
\end{algorithm}
We need the following corollary for the actual application of Theorem \ref{theo:unitester}:

\begin{corollary}\label{cor:unitester} There is an ``amplified'' version  of Algorithm \ref{alg:uni}
 that given $\epsilon>0$ and oracle access to a function $f:[n]\to [m]$ makes
$O(\frac{m^{1/3}\log\log m}{\epsilon^2})$ quantum queries and with probability at least $1- \frac{1}{\log^{2} m}$ outputs
 REJECT if $\norm{\PP_f- U}_1 \geq \epsilon$, and ACCEPT if $\norm{\PP_f - U}_\infty \leq \epsilon/4m$.
\end{corollary}

\begin{proof}[of Theorem~\ref{theo:unitester}]
Notice that Algorithm \ref{alg:uni} makes only $O(\frac{m^{1/3}}{\epsilon^2})$ queries:
$t = m^{1/3}$ classical queries are made initially, and the call to $\QEstimate$ requires additional
$O( m^{1/3}/\delta) = O(\frac{m^{1/3}}{\epsilon^2})$ queries.

Now we show that Algorithm \ref{alg:uni} satisfies the correctness conditions in Theorem \ref{theo:unitester}.
Let $V \subseteq [m]$ denote the multi-set of values $\{f(x) \mid x \in T\}$ (unlike $S$, the multi-set $V$ may contain some element of $[m]$ more than once).
If $\norm{\PP_f - U}_\infty \leq \epsilon/4m$ then $\PP_f(V)\leq (1+\frac{\epsilon}{4})t/m$, and hence
$$
p(t;m) \triangleq \Pr [\mathrm{the \ elements\ in\ } V \mathrm{\
are\ distinct}] \geq
\left(1-\frac{(1+\frac{\epsilon}{4})t}{m}\right)^t \geq 1-\frac{(1+\frac{\epsilon}{4})t^2}{m}
> 1 - o(1).
$$
Thus if $\norm{\PP_f - U}_\infty \leq \epsilon/4m$ then with probability at least $1-o(1)$,
the tester does not discover any collision. If, on the other hand, $\norm{\PP_f- U}_1 \geq
\epsilon$ and a collision is discovered, then the tester outputs
REJECT, as expected. Hence the following
lemma suffices for completing the proof of Theorem~\ref{theo:unitester}.

\begin{lemma}\label{lem:concentration}
Conditioned on the event that all elements in $V$ are distinct, we have
\begin{itemize}
\item if $\norm{\PP_f - U}_\infty \leq \epsilon/4m$ then $\Pr\Big[|\PP_f(V) - t/m| \leq
\frac{3\epsilon^2 t}{32m}\Big] \geq 1-o(1)$;

\item if $\norm{\PP_f-U}_1 \geq \epsilon$ then $\Pr\Big[|\PP_f(V) - t/m| > \frac{3\epsilon^2
t}{16m}\Big] \geq 1-o(1)$.
\end{itemize}
\end{lemma}
Assuming Lemma \ref{lem:concentration}, we first prove Theorem \ref{theo:unitester}.
Set $p \triangleq \PP_f(V)$, and recall that $t/m = 1/m^{2/3}$.

 If $\norm{\PP_f - U}_\infty \leq \epsilon/4m$ then
with probability at least $1-o(1)$ the elements in $V$ are
distinct and also $|p - 1/m^{2/3}| \leq
\frac{30\delta}{m^{2/3}}$. In this case, by Lemma \ref{lem:quantum}, with probability at least $5/6$
the estimate $p'$ computed by $\QEstimate$ satisfies
$|p-p'| \leq \frac{\delta \sqrt{p}}{m^{1/3}} + \frac{\delta^2}{m^{2/3}} \leq \frac{\delta \sqrt{(1+30\delta)/m^{2/3}}}{m^{1/3}}
+ \frac{\delta^2}{m^{2/3}} \leq \frac{2 \delta}{m^{2/3}}$, and by the triangle inequality $|p' - \frac{t}{m}| \leq 32 \delta\frac{t}{m}$.
Hence the overall probability that Algorithm \ref{alg:uni} outputs ACCEPT
is at least $5/6 - o(1) > 2/3$.

If $\norm{\PP_f- U}_1 \geq \epsilon$, then either Algorithm
\ref{alg:uni} discovers a collision and outputs REJECT, or
otherwise, $|p - 1/m^{2/3}| > \frac{60\delta}{m^{2/3}}$ with probability $1-o(1)$.
In the latter case, we make the following case distinction.
\begin{itemize}
\item {\bf Case $p \leq 10/m^{2/3}$:} By Lemma \ref{lem:quantum},
with probability at least $5/6$ the estimate $p'$ of $\QEstimate$ satisfies
$|p-p'| \leq \frac{\delta \sqrt{p}}{m^{1/3}} + \frac{\delta^2}{m^{2/3}} < \frac{10\delta}{m^{2/3}}$. Then by the triangle inequality,
$|p'-\frac{t}{m}| > \frac{60\delta}{m^{2/3}} - \frac{10\delta}{m^{2/3}} > 32 \delta\frac{t}{m}$.
\item {\bf Case $p > 10/m^{2/3}$:} In this case it is sufficient to prove that with probability at least $5/6$, $p' \geq p/2$ (which clearly implies
$|p'-\frac{t}{m}| > 32 \delta\frac{t}{m}$). This follows again by Lemma \ref{lem:quantum}, since $p > 10/m^{2/3}$ implies
$\frac{\delta \sqrt{p}}{m^{1/3}} + \frac{\delta^2}{m^{2/3}} \leq p/2$.
\end{itemize}
So the overall probability that Algorithm \ref{alg:uni} outputs REJECT
is at least $5/6 - o(1) > 2/3$.
\end{proof}

\

\begin{proof}[of Lemma~\ref{lem:concentration}]
Let $W_f(V) = \sum_{y \in V} \PP_f(y)$. Assuming that all elements
in $V$ are distinct, $\PP_f(V) = W_f(V)$. For the first item of the
lemma, it suffices to prove that if $\norm{\PP_f - U}_\infty \leq \epsilon/4m$ then
$$
\Pr\Big[|W_f(V) - \frac{t}{m}| > \frac{3\epsilon^2t}{32m}\Big] \leq
o(1)$$
and for the second item of the lemma, it suffices to prove that if $\norm{\PP_f-U}_1 \geq \epsilon$ then
$$
\Pr\Big[W_f(V) > (1 +\frac{3\epsilon^2}{16})\frac{t}{m}\Big] \geq
1- o(1).
$$
Note that the standard concentration inequalities cannot be used for proving the last inequality directly, because the
probabilities of certain elements under $\PP_f$ can be very high. To overcome this problem, we define $\tPP_f(y) \triangleq \min \{3/m,
\PP_f(y)\}$ and $\tW_f(V) \triangleq \sum_{y \in V} \tPP_f(y)$.
Clearly $\tW_f(V)  \leq W_f(V)$ for any $V$, hence proving $\Pr\Big[\tW_f(V) > (1 +\frac{3\epsilon^2}{16})\frac{t}{m}\Big] \geq
1- o(1)$ is sufficient. Surprisingly, this turns out to be easier:
\begin{lemma}\label{cl:conc} The following three statements hold
\begin{enumerate}
\item if $\norm{\PP_f - U}_\infty \leq \epsilon/4m$, then $\frac{t}{m} \leq  \Exp[\tW_f(V)] <
\left(1+\frac{\epsilon^2}{16}\right)\frac{t}{m}$
\item if $\norm{\PP_f-U}_1 \geq \epsilon$, then $\Exp[\tW_f(V)] >
\left(1+\frac{\epsilon^2}{4}\right)\frac{t}{m}$;
\item $\Pr\left[\Big|\tW_f(V) - \Exp[\tW_f(V)]\Big| > \frac{\epsilon^2t}{32m}\right] =
o(1)$.
\end{enumerate}
\end{lemma}
Assuming Lemma \ref{cl:conc} (it is proved in Section \ref{sec:claim-proof}) we have:
\begin{itemize}
 \item if $\norm{\PP_f - U}_\infty \leq \epsilon/4m$ then clearly $\tW_f(V) = W_f(V)$, therefore
$$
\Pr\Big[|W_f(V) - \frac{t}{m}| > \frac{3\epsilon^2t}{32m}\Big]
\leq \Pr\left[\Big|W_f(V) - \Exp[W_f(V)]\Big| >
\frac{\epsilon^2t}{32m}\right] = o(1);
$$
\item if $\norm{\PP_f-U}_1 \geq
\epsilon$ then

$$
\Pr\left[W_f(V) < (1 + \frac{3\epsilon^2}{16})\frac{t}{m}\right]
\leq \Pr\left[\tW_f(V) < (1 +
\frac{3\epsilon^2}{16})\frac{t}{m}\right]$$
$$ \leq \Pr\left[\Big|\tW_f(V) - \Exp[\tW_f(V)]\Big| >
\frac{\epsilon^2t}{16m}\right] \leq \Pr\left[\Big|\tW_f(V) -
\Exp[\tW_f(V)]\Big| > \frac{\epsilon^2t}{32m}\right] = o(1).
$$
\end{itemize}
Hence Lemma~\ref{lem:concentration} follows. 
\end{proof}

\subsubsection{Proof of Lemma \ref{cl:conc}}\label{sec:claim-proof}
We start by computing the expected value of $\tW_f(V)$.
$$
\Exp[\tW_f(V)] = \sum_{y\in V}\sum_{z\in [m]}\PP_f(z) \tPP_f(z) =
t\left(\sum_{z:\PP_f(z)< 3/m}\PP_f(z)^2 + \sum_{z:\PP_f(z)\geq 3/m}
3\PP_f(z)/m\right) $$ $$= t\left(\sum_{z\in [m]}\PP_f(z)^2 -
\sum_{z:\PP_f(z) \geq 3/m} \PP_f(z)(\PP_f(z) - 3/m)\right).
$$
Let $\delta(z) \triangleq \PP_f(z) - 1/m$ and let $r \triangleq |\{z \mid \delta(z)<2/m\}|$.
Then
$$
\Exp[\tW_f(V)] = t\left(\sum_{z\in [m]}(1/m + \delta(z))^2 -
\sum_{z:\delta(z) \geq 2/m} (1/m + \delta(z))(\delta(z) - 2/m)\right)
$$
and since $\sum_{z\in [m]} \delta(z) = 0$ we have
$$
= t\left(1/m + \sum_{z:\delta(z)<2/m}\delta(z)^2
 + 2(m-r)/m^2  + \sum_{z:\delta(z)\geq 2/m} \delta(z)/m \right)
$$

For the first item of the lemma, since $\delta(z) \leq \epsilon/4m$ we have $r = m$,
and hence the equality $W_f(V) = \tW_f(V)$ always holds as there are no $z$ for which
$\delta(z)\geq 2/m$. Therefore, from the above equation we have
$$
\Exp[W_f(V)] = t\left(1/m + \sum_{z:\delta(z)<2/m}\delta(z)^2\right) \geq \frac{t}{m}
$$
and
$$
\Exp[W_f(V)] = t\left(1/m + \sum_{z:\delta(z)<2/m}\delta(z)^2\right)
< t\left(1/m + \sum_{z:\delta(z)<2/m}(\epsilon/4m)^2\right) \leq
\left(1+\frac{\epsilon^2}{16}\right)\frac{t}{m}.
$$

Now we move to the second item of the lemma, where $\norm{\PP_f-U}_1 \geq \epsilon$.
By Cauchy-Schwarz we have
$$
\sum_{z:\delta(z)<2/m}\delta(z)^2 = \sum_{z:\delta(z)<2/m}|\delta(z)|^2 \geq \frac{1}{r}\Big(\sum_{z:\delta(z)<2/m}|\delta(z)|\Big)^2,
$$
hence
$$
\Exp[\tW_f(V)] \geq t\left(1/m +\frac{1}{r}\Big(\sum_{z:\delta(z)<2/m}|\delta(z)|\Big)^2 + \frac{1}{m}\sum_{z:\delta(z)\geq
2/m}\delta(z)\right)$$ $$ \geq \frac{t}{m}\left(1 + \Big(\sum_{z:\delta(z)<2/m}|\delta(z)|\Big)^2 + \sum_{z:\delta(z)\geq
2/m}\delta(z)\right).
$$
Since $\sum_{z \in [m]} |\delta(z)|=\norm{\PP_f-U}_1 \geq \epsilon$, at least one of
$$
\sum_{z:\delta(z)<2/m}|\delta(z)| > \epsilon/2
$$
or
$$
\sum_{z:\delta(z)\geq 2/m}|\delta(z)| = \sum_{z:\delta(z)\geq 2/m}\delta(z) \geq \epsilon/2
$$
must hold. In both cases we have $\Exp[\tW_f(V)] > \frac{t}{m}(1 + \frac{\epsilon^2}{4})$, as required.

Finally, we prove the third statement of the lemma. By Hoeffding's Inequality we have
$$
\Pr\left[\Exp[\tW_f(V)] - \tW_f(V) >
\frac{\epsilon^2t}{32m}\right] \leq \exp\left(-\frac{2\epsilon^4
t^2}{1024m^2\sum_{y\in V}(b_y - a_y)^2}\right),
$$
where $b_y$ and $a_y$ are upper and lower bounds on $\tPP(y)$.
Since $b_y \leq 3/m$ and $a_y \geq 0$ for all $y\in[m]$, we get
$$
\Pr\left[\Exp[\tW_f(V)] - \tW_f(V) > \frac{\epsilon^2t}{32m}\right]
\leq \exp(-\Omega(\epsilon^4 t)) = o(1).
$$

\subsection{Testing Closeness to a Known Distribution}\label{sec:general}
In this section we prove Theorem \ref{thm:1unknown} based on Theorem
\ref{theo:unitester}. Let $\PP_f$ be an unknown distribution and let
$\PP_g$ be a known distribution, defined by $f,g:[n] \to [m]$
respectively. We show that for any $\epsilon > 0$, Algorithm
\ref{alg:ku} makes $O(\frac{m^{1/3} \log^2 m \log \log m}{\epsilon^5})$ queries and
distinguishes the case $\norm{\PP_f-\PP_g}_1=0$ from the case $\norm{\PP_f-
\PP_g}_1 > 5\epsilon$ with probability at least $2/3$, satisfying
the requirements of Theorem \ref{thm:1unknown}.\footnote{We use $5\epsilon$ instead $\epsilon$ for better readability in the sequel.}

\begin{algorithm}
\caption{(Tests closeness to a known distribution.)}
\label{alg:ku}
\begin{algorithmic}[1]
\STATE let $\mathcal{M} \triangleq \{M_0, \dots, M_k\} \gets
Bucket(\PP_g, [m], \frac{\epsilon}{4})$ for $k=\frac{2 \log
m}{\log(1+\epsilon/4)}$ \FOR {$i=1$ to $k$}
    \IF {$\PP_g(M_i) \geq \epsilon/k$}
        \IF {$\norm{({\PP_f})_{|M_i}-U_{|M_i}}_1 \geq \epsilon$ (check using the amplified version of Algorithm~\ref{alg:uni} from Corollary \ref{cor:unitester})}
            \STATE REJECT \label{line:rej1}
        \ENDIF
    \ENDIF
\ENDFOR \IF {$\norm{(\PP_f)_{\langle \mathcal{M} \rangle} -
(\PP_g)_{\langle \mathcal{M} \rangle}}_1 > \epsilon/4$
(check classically with $O(\sqrt{k}) = O(\log m)$ queries \cite{BFF+01})}
    \STATE REJECT \label{line:rej2}
\ENDIF
\STATE ACCEPT
\end{algorithmic}
\end{algorithm}

Observe that no queries are made by Algorithm \ref{alg:ku} itself,
and the total number of queries made by calls to Algorithm
\ref{alg:uni} is bounded by $k \cdot O(\frac{k}{\epsilon} \cdot \frac{m^{1/3} \log \log m}{\epsilon^2}) +
O(\sqrt{k}) = O(\frac{m^{1/3} \log^2 m\log \log m}{\epsilon^5})$.\footnote{The additional factor of $\frac{k}{\epsilon}$ is for executing 
Algorithm \ref{alg:uni} on the conditional distributions $({\PP_f})_{|M_i}$, with $\PP_f(M_i) \geq \frac{\epsilon}{k}$.}
 In addition, the failure probability of Algorithm \ref{alg:uni} is at most
$1/\log^{2}m \ll 1/k$, so we can assume that with high probability none
of its executions failed.

For any $i\in [k]$ and any $x \in M_i$, by the definition of the
buckets $\frac{(1 + \epsilon/4)^{i-1}}{m\log m} \leq \PP_g(x) \leq
\frac{(1 + \epsilon/4)^{i}}{m\log m}$. Thus, for any $i \in [k]$ and
$x\in M_i$, $(1 - \frac{\epsilon}{4})/|M_i| < 1/(1 +
\frac{\epsilon}{4})|M_i| <  (\PP_g)_{|M_i}(x) < (1 +
\frac{\epsilon}{4})/|M_i|$, or equivalently for any $i \in [k]$ we
have $\norm{(\PP_g)_{|M_i} - U_{|M_i}}_\infty \leq
\frac{\epsilon}{4|M_i|}$. This means that if $\norm{\PP_f - \PP_g}_1
= 0$ then
\begin{enumerate}
 \item for any $i \in [k]$, $\norm{(\PP_f)_{|M_i} - U_{|M_i}}_\infty \leq \frac{\epsilon}{4|M_i|}$
and thus the tester never outputs REJECT in Line \ref{line:rej1} (since we assumed that Algorithm \ref{alg:uni} did not err in any of its executions).
\item $\norm{(\PP_f)_{\langle \mathcal{M} \rangle} -
(\PP_g)_{\langle \mathcal{M} \rangle}}_1 = 0$, and hence the tester does not output REJECT in Line \ref{line:rej2} either.
\end{enumerate}

On the other hand, if $\norm{\PP_f - \PP_g}_1 > 5\epsilon$ then by
Corollary~\ref{coro:partition} we know that either $|(\PP_f)_{\langle
\mathcal{M} \rangle} - (\PP_g)_{\langle \mathcal{M} \rangle}| >
\epsilon/4$ or there is at least one $i\in [k]$ for which
$\PP_f(M_i) \geq \epsilon/k$ and $\norm{(\PP_f)_{|M_i} -
(\PP_g)_{|M_i}}_1 > 5\epsilon/4$ (otherwise $\norm{\PP_f - \PP_g}_1$
must be smaller than $2( 5\epsilon/4 + \epsilon/4 +  \epsilon) =
5\epsilon$). In the first case the tester will reject in Line
\ref{line:rej2}. In the second case the tester will reject in Line
\ref{line:rej1} as $\norm{(\PP_f)_{|M_i} - (\PP_g)_{|M_i}}_1 >
5\epsilon/4$ implies (by the triangle inequality)
$\norm{(\PP_f)_{|M_i} - U_{|M_i}}_1 > \epsilon$, since
$\norm{(\PP_g)_{|M_i} - U_{|M_i}}_1 < \epsilon/4$ by
Lemma~\ref{lem:bucket}.

\section{Quantum Lower Bounds for Testing Distributions} \label{sec:qlbs}

Here we show that our quantum testing algorithm for the
known-unknown case is close to optimal: even for testing an unknown
distribution (given as $f:[n]\to[m]$) against the uniform one, we
need $\Omega\Big(m^{1/3}\Big)$ quantum queries. As Bravyi, Hassidim, and
Harrow~\cite{BHH10} also independently observed, such a lower bound can be derived
from known lower bounds for the collision problem.  However, one has
to be careful to use the version of the lower bound that applies to
functions $f:[m]\to[m]$, due to Ambainis~\cite{ambainis:colsmallrange} and
Kutin~\cite{kutin:collision}, rather than the earlier lower bound of
Aaronson and Shi~\cite{aaronson&shi:collision} that had to assume a
larger range-size.

\begin{theorem}
Let $A$ be a quantum algorithm that given a fixed $\epsilon\in[0,1]$ tests whether an unknown distribution
is equal to uniform or at least $\epsilon$-far from it,
meaning that for every $f:[n]\to[m]$, with success probability at least $2/3$, it decides whether
$\PP_f=U$ or $\norm{\PP_f - U}_1\geq\epsilon$ (under the promise that one of these two cases holds).
Then $A$ makes $\Omega\Big(m^{1/3}\Big)$ queries to~$f$.
\end{theorem}

\begin{proof}
Consider the following distribution on $f:[m] \to [m]$:
with probability 1/2, $f$ is a random 1-1 function (equivalently, a random permutation on $[m]$),
and with probability 1/2, $f$ is a random 2-to-1 function.
In the first case we have $\PP_f=U$, while in the second case
$\PP_f(j)\in\{0,2/m\}$ for all $j\in[m]$ and hence $\norm{\PP_f - U}_1=1$.
Thus a quantum testing algorithm like $A$ can decide between these two cases with high success probability.
But Ambainis~\cite{ambainis:colsmallrange} and Kutin~\cite{kutin:collision} showed that this requires $\Omega(m^{1/3})$ queries.
%
\end{proof}

\section{Quantum Lower Bounds for Reconstructing Distributions}\label{sec:reconstructionlower}

Previously we studied the problem of \emph{deciding} whether an
unknown distribution, given by $f:[n] \to [m]$, is close to or far
from another distribution (which itself may be known or unknown). Of
course, the easiest way to solve such a decision problem would be to
\emph{reconstruct} the unknown distribution, up to some small
$L_1$-error. Efficiently solving the reconstruction problem, say in
$m^{1/2}$ or even $m^{1/3}$ queries, would immediately allow us to
solve the decision problem. However, below we prove that even
quantum algorithms cannot solve the reconstruction problem efficiently.

\begin{theorem} \label{theo:rec}
Let $0<\epsilon<1/2$ be a fixed constant.
Let $A$ be a quantum algorithm that solves the reconstruction problem, meaning that
for every $f:[n]\to[m]$, with probability at least $2/3$, it outputs a probability
distribution $\PP\in[0,1]^m$ such that $\norm{\PP - \PP_f}_1\leq\epsilon$.
Then $A$ makes $\Omega(m/\log m)$ queries to~$f$.
\end{theorem}

\begin{proof}
The proof uses some basic quantum information theory, and is most easily stated in a communication setting.
Suppose Alice has a uniformly distributed $m$-bit string $x$ of weight $m/2$.
This is a random variable with entropy $\log\binom{m}{m/2}=m-O(\log m)$ bits.
Let $q$ be the number of queries $A$ makes.
We will show below that Alice can give Bob $\Omega(m)$ bits of information (about $x$),
by a process that (interactively) communicates $O(q\log m)$ qubits.
By Holevo's Theorem~\cite{holevo} (see also~\cite[Theorem~2]{cdnt:ip}),
establishing $k$ bits of mutual information requires communicating at least $k$ qubits,
hence $q=\Omega(m/\log m)$.

Given an $x\in\{0,1\}^m$ of weight $n=m/2$, let $f:[n]\to[m]$ be an injective function to $\{j \mid x_j=1\}$,
and let $\PP_f$ be the corresponding probability distribution over $m$ elements
(which is $\PP_f(j)=2/m$ where $x_j=1$, and $\PP_f(j)=0$ where $x_j=0$).
Let $\PP$ be the distribution output by algorithm $A$ on $f$.
We have $\norm{\PP - \PP_f}_1\leq\epsilon$ with probability at least $2/3$.
Define a string $\widetilde{x}\in\{0,1\}^m$ by $\widetilde{x}_j=1$ iff $\PP(j)\geq 1/m$.
Note that at each position $j\in[m]$ where $x_j\neq \widetilde{x}_j$,
we have $|\PP(j)-\PP_f(j)|\geq 1/m$. Hence $\norm{\PP - \PP_f}_1\geq d(x,\widetilde{x})/m$.
Since $\norm{\PP - \PP_f}_1\leq\epsilon$ (with probability at least $2/3$),
the algorithm's output allows us to produce (with probability at least $2/3$)
a string $\widetilde{x}\in\{0,1\}^m$ at Hamming distance $d(x,\widetilde{x})\leq\epsilon m$ from $x$.
But then it is easy to calculate that the mutual information between $x$ and $\widetilde{x}$ is $\Omega(m)$ bits.

Finally, to put this in the communication setting, note that Bob can run the algorithm $A$,
implementing each query to $f$ by sending the $O(\log n)$-qubit query-register to Alice,
who plugs in the right answer and sends it back (this idea comes from~\cite{BuhrmanCleveWigderson98}).
The overall communication is $O(q\log m)$ qubits.
\end{proof}

\section{From Sampling Problems to Oracle Problems}\label{sec:sampling}

A standard way to access a probability distribution $\PP$ on $[m]$
is by \emph{sampling} it: sampling once gives the outcome $y\in[m]$
with probability $\PP(y)$. However, in this paper we usually assume
that we can access the distribution by querying a function
$f:[n]\to[m]$, where the probability of $y$ is now interpreted as
the fraction of the domain that is mapped to $y$. Below we describe
the connection between these two approaches.

Suppose we sample $\PP$ $n$ times, and estimate each probability
$\PP(y)$ by the fraction $\widetilde{\PP}(y)$ of times $y$ occurs
among the $n$ outcomes. We will analyze how good an estimator this
is for $\PP(y)$. For all $j\in[n]$, let $Y_j$ be the indicator
random variable that is 1 if the $j$th sample is $y$, and 0
otherwise. This has expectation $\Exp[Y_j]=\PP(y)$ and variance
$\Var[Y_j]=\PP(y)(1-\PP(y))$. Our estimator is
$\widetilde{\PP}(y)=\sum_{j\in[n]} Y_j/n$. This has expectation
$\Exp[\widetilde{\PP}(y)]=\PP(y)$ and variance
$\Var[\widetilde{\PP}(y)]=\PP(y)(1-\PP(y))/n$, since the $Y_j$'s are
independent. Now we can bound the expected error of our estimator
for $\PP(y)$ by
$$
\Exp\left[|\widetilde{\PP}(y) - \PP(y)|\right]
\leq\sqrt{\Exp\left[|\widetilde{\PP}(y) - \PP(y)|^2\right]}
=\sqrt{\Var\left[\widetilde{\PP}(y)\right]}
\leq\sqrt{\PP(y)/n}.
$$
And we can bound the expected $L_1$-distance between the original distribution $\PP$
and its approximation $\widetilde{\PP}$ by
\begin{eqnarray*}
\Exp\left[\norm{\widetilde{\PP} - \PP}_1\right]
= \sum_{y\in[m]}\Exp\left[|\widetilde{\PP}(y) - \PP(y)|\right]
\leq\sum_{y\in[m]}\sqrt{\PP(y)/n}
\leq \sqrt{m/n},
\end{eqnarray*}
where the last inequality used Cauchy-Schwarz and the fact that
$\sum_y\PP(y)=1$. For instance, if $n=10000 m$ then
$\Exp[\norm{\widetilde{\PP} - \PP}_1]\leq 1/100$, and hence (by
Markov's Inequality) $\norm{\widetilde{\PP} - \PP}_1\leq 1/10$ with
probability at least 9/10. If we now define a function $f:[n]\to[m]$
by setting $f(j)$ to the $j$th value in the sample, we have obtained
a representation which is a good approximation of the original
distribution. Note that if $n=o(m)$ then we cannot hope to be able
to approximately represent all possible $m$-element distributions by
some $f:[n]\to[m]$, since all probabilities will be integer
multiples of $1/n$.  For instance if $\PP$ is uniform and $n=o(m)$,
then the total $L_1$-distance between $\PP$ and a $\widetilde{\PP}$
induced by any $f:[n]\to[m]$ is near-maximal. Accordingly, the
typical case we are interested in is $n=\Theta(m)$.

\section{Proof of Theorem \ref{theo:uu_gi}} \label{sec:uu_gi}

In \cite{FM}, the bottleneck (with respect to the query complexity) of the algorithm for testing
graph isomorphism in the known-unknown case is the subroutine that
tests closeness between two distributions over $V$.
All other parts of the
algorithm make only a polylogarithmic number of queries. Therefore, our main theorem
implies that with quantum oracle access, graph isomorphism in the known-unknown setting can be tested
with $\widetilde{O}(|V|^{1/3})$ queries.

On the other hand, a general lower bound on the query complexity of testing distributions in
the known-unknown case need not imply a lower bound for testing graph
isomorphism. But still, in \cite{FM} it is proved that a lower bound on the query complexity for deciding
whether the function $f:[n]\to [n]$ is one-to-one (i.e., injective) or is two-to-one 
(i.e., the pre-image of any $j\in [n]$ is either empty or of size $2$) 
is sufficient for showing a matching lower bound for graph isomorphism.
Since our quantum lower bound for the known-unknown testing case is derived
from exactly that problem (see Section~\ref{sec:qlbs}), we get a matching lower bound
of $\Omega(|V|^{1/3})$ on the number of quantum queries necessary
for testing graph isomorphism in the known-unknown case.

For the unknown-unknown case,
the lower bound mentioned in Theorem \ref{theo:uu_gi} follows from
 the lower bound for the known-unknown case. To get the upper bound of $\widetilde{O}(|V|^{7/6})$ queries, we have to slightly modify the algorithm from \cite{FM}. We start by outlining the ideas in the algorithm of \cite{FM} for testing isomorphism between two unknown graphs $G$ and $H$.

Let $G$ be a graph and $C_G \subseteq V(G)$. A {\em $C_G$-label} of
a vertex $v \in V(G)$ is a binary vector of length $|C_G|$ that
represents the neighbors of $v$ in $C_G$. The distribution
$\PP_{C_G}$ over $\{0,1\}^{|C_G|}$ is defined according to the graph
$G$, where for every $x \in \{0,1\}^{|C_G|}$ the probability
$\PP_{C_G}(x)$ is proportional to the number of vertices in $G$ with
$C_G$-label equal to $x$. Notice that the support of $\PP_{C_G}$ is
bounded by $|V(G)|$.

The algorithm of~\cite{FM} is based on two main observations:
\begin{enumerate}
\item if there is an isomorphism $\sigma$ between $G$ and $H$, then for every $C_G \subseteq V(G)$ and the corresponding $C_H \triangleq \sigma(C_G)$, the distributions $\PP_{C_G}$ and $\PP_{C_H}$ are identical.
\item if  $G$ and $H$ are far from being isomorphic, then for every equally-sized (and not too small) $C_G \subseteq V(G)$ and $C_H \subseteq V(H)$, either the distributions $\PP_{C_G}$ and $\PP_{C_H}$ are far, or otherwise it is possible to test with only a poly-logarithmic number of queries that there exists no isomorphism that maps $C_G$ to $C_H$.
\end{enumerate}
Once these observations are made, the high level idea in the algorithm of $\cite{FM}$ is to go over a sequence of pairs of sets $C_G,C_H$ (such that with high probability at least one of them satisfies $C_H \triangleq \sigma(C_G)$ if indeed an isomorphism $\sigma$ exists), and to test closeness between the corresponding distributions $\PP_{C_G}$ and $\PP_{C_H}$.

This sequence of pairs is defined as follows: first we pick (at random) a set $U_G$ of $|V|^{1/4} \log^3 |V|$ vertices
from $G$ and a set $U_H$ of $|V|^{3/4} \log^3 |V|$ vertices from $H$.
Then we make all $|V|^{5/4} \log^3 |V|$ possible queries in $U_G
\times V(G)$. After this, for any $C_G \subseteq U_G$ the distribution $\PP_{C_G}$ is known exactly.
Indeed, the sequence of sets $C_G,C_H$ will consist of all pairs $C_G \subseteq U_G,C_H \subseteq U_H$, where both $C_G$ and $C_H$ are of size $\log^2 |V|$.
It is not hard to prove that if $G$ and $H$ have an isomorphism $\sigma$,
then with probability $1-o(1)$
the size of $U_H \cap \sigma(U_G)$ will exceed $\log^2 |V|$, and hence one of the pairs will satisfy $C_H \triangleq \sigma(C_G)$.

Now, for each pair $C_G,C_H$ we test if the distributions $\PP_{C_G}$ and $\PP_{C_H}$ are identical. Since we know the distributions
$\PP_{C_G}$ (for every $C_G \subseteq U_G$), we only need to sample the distributions $\PP_{C_H}$.
Sampling the distributions $\PP_{C_H}$ is done by taking a set $S \subseteq V(H)$ of size $\widetilde{O}(\sqrt{|V|})$ and re-using it for all these tests. In total, the algorithm in \cite{FM}  makes roughly $|U_G
\times V(G)| + |U_H \times S| = \widetilde{O}(|V|^{5/4})$ queries.

To get the desired improvement, we follow the same path, but use our quantum distribution tester instead of the classical one. This allows us to reduce the size of the set $S$ to  $\widetilde{O}(|V|^{1/3})$. Consequently, in order to balance the amount of queries we make in both graphs, we will resize the sets $U_G$ and $U_H$ to $\widetilde{O}(|V|^{1/6})$ and $\widetilde{O}(|V|^{5/6})$ respectively, which still satisfies the ``large-intersection'' property and brings the total number of queries down to  $|U_G
\times V(G)| + |U_H \times S| =\widetilde{O}(|V|^{7/6})$.

\section{Proof of Theorem \ref{thperiodicity}} \label{sec:periodicity}
\subsection{Quantum Upper Bound}

The quantum tester is very simple, and completely based on existing ideas.
First, run a variant of Shor's algorithm to find the period of $f$ (if there is one), using $O(1)$ queries.
Second, test whether the purported period is indeed the period, using another $O(1)$ queries as described above.
Accept iff the latter test accepts.

For the sake of completeness we sketch here how Shor's algorithm can be used to find the unknown period $p$ of
an $f$ that is promised to be 1-1-$p$-periodic for some value of $p\leq\sqrt{n}/2$.
Here is the algorithm:%
\footnote{For this to work, the 1-1 property on $[p]$ is crucial; for instance, quantum algorithms need
about $\sqrt{n}$ queries to find the period of functions with range $\zo$.
Also the fact that $p=O(\sqrt{n})$ is important, because the quantum algorithm needs to see many repetitions
of the period on the domain $[n]$.}
\begin{enumerate}
\item First prepare the 2-register quantum state
$
\displaystyle\frac{1}{\sqrt{n}}\sum_{i\in[n]}\ket{i}\ket{0}
$
\item Query $f$ once (in superposition), giving
$
\displaystyle\frac{1}{\sqrt{n}}\sum_{i\in[n]}\ket{i}\ket{f(i)}
$
\item Measure the second register, which gives some $f(s)$ for $s\in[p]$ and collapses
the first register to the~$i$ having the same $f$-value:
$
\displaystyle \frac{1}{\sqrt{\lfloor n/p\rfloor}}\sum_{i\in[n], i=s {\rm \ mod\ } p}\ket{i}\ket{f(i)}
$
\item Do a quantum Fourier transform\footnote{This is the unitary map
$\ket{x}\to \frac{1}{\sqrt{n}}\sum_{y\in[n]}e^{2\pi i xy/n}\ket{y}$.
If $n$ is a power of 2 (which we can assume here without loss of generality),
then the QFT can be implemented using $O((\log n)^2)$ elementary quantum gates~\cite[Section~5.1]{nielsen&chuang:qc}.}
on the first register and measure.\\
Some analysis shows that with high probability the measurement gives an $i$ such that
$\displaystyle
\left|\frac{i}{n}-\frac{c}{p}\right|<\frac{1}{2n}
$,
where $c$ is a random (essentially uniform) integer in $[p]$.
Using continued fraction expansion, we can then calculate the unknown fraction $c/p$ from the known fraction $i/n$.%
\footnote{Two distinct fractions each with denominator $\leq \sqrt{n}/2$ are at least $4/n$ apart.
Hence there is only one fraction with denominator at most $\sqrt{n}/2$ within distance $2/n$ from the known fraction $i/n$.
This unique fraction can only be $c/p$, and CFE efficiently finds it for us.
Note that we do not obtain $c$ and $p$ separately, but just their ratio given as a numerator and
a denominator in lowest terms.
If $c$ and $p$ were coprime that would be enough, but that need not happen with high probability.}
\item Doing the above 4 steps $k$ times gives fractions $c_1/p,\ldots,c_k/p$, each given as a numerator and a denominator (in lowest terms).
Each of the $k$ denominators divides $p$, and if $k$ is a sufficiently large constant
then with high probability (over the $c_i$'s), their least common multiple is $p$.
\end{enumerate}

\subsection{Classical Lower Bound}

We saw above that quantum computers can efficiently test {\sc 1-1-periodicity} $\PP_{\sqrt{n}/4,\sqrt{n}/2}$.
Here we will show that this is not the case for classical testers: those need roughly $\sqrt{r}$ queries
for 1-1-periodicity testing $\PP_{r/2,r}$, in particular roughly $n^{1/4}$ queries for $r=\sqrt{n}/2$.
Our proof follows along the lines of Lachish and Newman~\cite{lachish&newman:periodicity}.
However, since their proof applies to functions with range 0/1 that need not satisfy the 1-1 property, some modifications are needed.

Fix a sufficiently large even integer $r<n/2$.
We will use Yao's principle, proving a lower bound for \emph{deterministic} query testers with error probability $\leq 1/3$
in distinguishing two distributions, one on negative instances and one on positive instances.
First, the ``negative'' distribution ${\cal D}_N$ is uniform on all $f:[n]\to [m]$ that are $\epsilon$-far from $\PP_{r/2,r}$.
Second, the ``positive'' distribution ${\cal D}_P$ chooses a \emph{prime} period $p\in[r/2,r]$ uniformly,
then chooses a 1-1 function $[p]\to [m]$ uniformly (equivalently, chooses a sequence of $p$ distinct elements from $[m]$),
and then completes $f$ by repeating this period until the domain $[n]$ is ``full''.
Note that the last period will not be completed if $p\not|n$.

Suppose $q=o(\sqrt{r/\log r\log n})$ is the number of queries of our deterministic tester.
Fix a set $Q=\{i_1,\ldots,i_q\}\subseteq[n]$ of $q$ queries.
Let $f(Q)\in[m]^q$ denote the concatenated answers $f(i_1),\ldots,f(i_q)$.
We prove two lemmas, one for the negative and one for the positive distribution,
showing $f(Q)$ to be close to uniformly distributed in both cases.

\begin{lemma}
For all $\eta\in[m]^q$, we have $\Pr_{{\cal D}_N}[f(Q)=\eta]=(1\pm o(1))m^{-q}$.
\end{lemma}

\begin{proof}
We first upper bound the number of functions $f:[n]\to [m]$ that are $\epsilon$-\emph{close} to $p$-periodic for a specific $p$.
The number of functions that are perfectly $p$-periodic is $m^p$, since such a function is determined by its first $p$ values.
The number of functions $\epsilon$-close to a fixed $f$ is at most ${n \choose \epsilon n}m^{\epsilon n}$.
Hence the number of functions $\epsilon$-close to $\PP_p$ is at most $m^p{n\choose \epsilon n}m^{\epsilon n}$.
Therefore, under the uniform distribution $\cal U$ on all $m^n$ functions $f:[n]\to [m]$,
the probability that there is a period $p\leq r$ for which $f$ is $\epsilon$-close to $\PP_p$ is at most
$$
\frac{r\cdot m^r{n\choose \epsilon n}m^{\epsilon n}}{m^n}\leq m^{n/2+H(\epsilon)n/\log m+\epsilon n - n},
$$
where we used $r<n/2$, $n\leq m$, and ${n\choose\epsilon n}\leq 2^{H(\epsilon)n}$
with $H(\cdot)$ denoting binary entropy. If $\epsilon$ is a sufficiently small constant,
then this probability is $o(m^{-q})$ (in fact much smaller than that).
Hence the variation distance between ${\cal D}_N$ and the uniform distribution $\cal U$ is $o(m^{-q})$,
and we have
$$
\left|\Pr_{{\cal D}_N}[f(Q)=\eta]-m^{-q}\right| = \left|\Pr_{{\cal D}_N}[f(Q)=\eta]-\Pr_{\cal U}[f(Q)=\eta]\right| = o(m^{-q}).
$$
\end{proof}

\begin{lemma}
There exists an event $B$ such that $\Pr_{{\cal D}_P}[B]=o(1)$,
and for all $\eta\in[m]^q$ with distinct coordinates,
we have $\Pr_{{\cal D}_P}[f(Q)=\eta\mid \overline{B}]=(1\pm o(1))m^{-q}$.
\end{lemma}

\begin{proof}
The distribution ${\cal D}_P$ uniformly chooses a prime period $p\in[r/2,r]$.
By the prime number theorem (assuming $r$ is at least a sufficiently large constant,
which we may do because the lower bound is trivial for constant $r$),
the number of distinct primes in this interval is asymptotically
$$
\frac{r}{\ln(r)}-\frac{r/2}{\ln(r/2)} \geq \frac{r}{2\log r}.
$$
Let $B$ be the event that a $p$ is chosen for which there exist
distinct $i,j\in Q$ satisfying $i=j$ mod $p$ (equivalently, $p$ divides $i-j$).
For each fixed $i,j$ there are at most $\log n$ primes dividing $i-j$.
Hence at most ${q\choose 2}\log n=o(r/\log r)$ $p$'s out of the at least $r/2\log r$ possible $p$'s can cause event $B$,
implying $\Pr_{{\cal D}_P}[B]=o(1)$.

Conditioned on $B$ not happening, $f(Q)$ is a uniformly random element of $[m]^q$ with distinct coordinates,
hence for each $\eta\in[m]^q$ with distinct coordinates we have
$$
\Pr_{{\cal D}_P}[f(Q)=\eta \mid \overline{B}]=\frac{1}{m}\frac{1}{m-1}\cdots\frac{1}{m-q+1}
=m^{-q}\prod_{i=0}^{q-1}\left(1+\frac{i}{m-i}\right)=(1+o(1))m^{-q}.
$$
\end{proof}

Since $(1-o(1))m^q$ of all $\eta\in[m]^q$ have distinct coordinates,
their weight under ${\cal D}_P$ sums to $1-o(1)$, and the other possible $\eta$ comprise
only a $o(1)$-fraction of the overall weight.
The query-answers $f(Q)$ are the only access the algorithm has to the input. Hence the previous two lemmas
imply that an algorithm with $o(\sqrt{r/\log r\log n})$ queries cannot distinguish ${\cal D}_P$ and ${\cal D}_N$
with probability better than $1/2+o(1)$.  This establishes the claimed classical lower bound.

\section{Summary and Open Problems}
In this paper we studied and compared the quantum and classical query complexities of
a number of testing problems.
The first problem is deciding whether two probability distributions on a set $[m]$ are equal or
$\epsilon$-far. Our main result is a quantum tester for the case
where one of the two distributions is known (i.e., given explicitly)
while the other is unknown and represented by a function that can be
queried. Our tester uses roughly $m^{1/3}$ queries to the function,
which is essentially optimal. It would be very interesting to extend
this quantum upper bound to the case where \emph{both} distributions are
unknown.  Such a quantum tester would show that the known-unknown
and unknown-unknown cases have the same complexity in the quantum
world. In contrast, they are known to have different complexities in
the classical world: about $m^{1/2}$ queries for the known-unknown
case and about $m^{2/3}$ queries for the unknown-unknown case.
The classical counterparts of these tasks play an important
role in many problems related to property testing. We already mentioned one example, the graph isomorphism problem, where
distribution testers are used as a black-box. 
We hope that the quantum analogues developed here and in~\cite{BHH10} will find similar use.

The second testing problem is deciding whether a given function $f:[n]\to [m]$ is periodic or far from periodic.
For the specific version of the problem that we considered (where in the first case the period
is at most about $\sqrt{n}$, and the function is injective within each period), we proved that quantum
testers need only a constant number of queries (using Shor's algorithm), while classical algorithms
need about $n^{1/4}$ queries.  Both this result and Aaronson's recent result 
on ``Fourier checking''~\cite{aaronson:bqpph} contrast 
with the constant-vs-$\log n$ and $\log n$-vs-$\sqrt{n}$
separations obtained by Buhrman et al.~\cite{bfnr:qprop} for other testing problems,
but still leave open their question: is there a testing problem where the separation is ``maximal'',
in the sense that quantum testers need only $O(1)$ queries while classical testers need $\Omega(n)$?

\subsubsection*{Acknowledgements}
We thank Avinatan Hassidim, Harry Buhrman and Prahladh Harsha for useful discussions, Frederic Magniez for a reference to~\cite{fmsp:testhidden}, and Scott Aaronson for pointing out that his Fourier checking result in~\cite{aaronson:bqpph} was the first constant-vs-polynomial quantum speed-up in property testing.


\begin{thebibliography}{10}

\bibitem{aaronson:bqpph}
S.~Aaronson.
\newblock {BQP} and the {P}olynomial {H}ierarchy.
\newblock In {\em Proceedings of 42nd ACM STOC}, 2010 (to appear).
\newblock arXiv:0910.4698.

\bibitem{aaronson&shi:collision}
S.~Aaronson and Y.~Shi.
\newblock Quantum lower bounds for the collision and the element distinctness
  problems.
\newblock {\em Journal of the ACM}, 51(4):595--605, 2004.

\bibitem{ambainis:colsmallrange}
A.~Ambainis.
\newblock Polynomial degree and lower bounds in quantum complexity: Collision
  and element distinctness with small range.
\newblock {\em Theory of Computing}, 1(1):37--46, 2005.
\newblock quant-ph/0305179.

\bibitem{ambainis:edj}
A.~Ambainis.
\newblock Quantum walk algorithm for element distinctness.
\newblock {\em SIAM Journal on Computing}, 37(1):210--239, 2007.
\newblock Earlier version in FOCS'04. quant-ph/0311001.

\bibitem{acrsz:andor}
A.~Ambainis, A.~Childs, B.~Reichardt, R.~{\v{S}}palek, and S.~Zhang.
\newblock Any {AND-OR} formula of size $n$ can be evaluated in time
  {$N^{1/2+o(1)}$} on a quantum computer.
\newblock In {\em Proceedings of 48th IEEE FOCS}, 2007.

\bibitem{atici&servedio:testing}
A.~Atici and R.~Servedio.
\newblock Quantum algorithms for learning and testing juntas.
\newblock {\em Quantum Information Processing}, 6(5):323--348, 2009.

\bibitem{BFF+01}
T.~Batu, L.~Fortnow, E.~Fischer, R.~Kumar, R.~Rubinfeld, and P.~White.
\newblock Testing random variables for independence and identity.
\newblock In {\em Proceedings of 42nd IEEE FOCS}, pages 442--451, 2001.

\bibitem{BFR+00}
T.~Batu, L.~Fortnow, R.~Rubinfeld, W.~D. Smith, and P.~White.
\newblock Testing that distributions are close.
\newblock In {\em Proceedings of 41st IEEE FOCS}, pages 259--269, 2000.

\bibitem{blr:selftest}
M.~Blum, M.~Luby, and R.~Rubinfeld.
\newblock Self-testing/correcting with applications to numerical problems.
\newblock {\em Journal of Computer and System Sciences}, 47(3):549--595, 1993.
\newblock Earlier version in STOC'90.

\bibitem{bhmt:countingj}
G.~Brassard, P.~H{\o}yer, M.~Mosca, and A.~Tapp.
\newblock Quantum amplitude amplification and estimation.
\newblock In {\em Quantum Computation and Quantum Information: A Millennium
  Volume}, volume 305 of {\em AMS Contemporary Mathematics Series}, pages
  53--74. 2002.
\newblock quant-ph/0005055.

\bibitem{BHH10}
S.~Bravyi, A.~Hassidim, and A.~Harrow.
\newblock Quantum algorithms for testing properties of distributions.
\newblock In {\em Proceedings of 27th Annual Symposium on Theoretical Aspects
  of Computer Science (STACS'2010)}, 2010.
\newblock abs/0907.3920.

\bibitem{BuhrmanCleveWigderson98}
H.~Buhrman, R.~Cleve, and A.~Wigderson.
\newblock Quantum vs.~classical communication and computation.
\newblock In {\em Proceedings of 30th ACM STOC}, pages 63--68, 1998.
\newblock quant-ph/9802040.

\bibitem{bfnr:qprop}
H.~Buhrman, L.~Fortnow, I.~Newman, and H.~R{\"o}hrig.
\newblock Quantum property testing.
\newblock In {\em Proceedings of 14th ACM-SIAM SODA}, pages 480--488, 2003.
\newblock quant-ph/0201117.

\bibitem{buhrman&wolf:dectreesurvey}
H.~Buhrman and R.~{de} Wolf.
\newblock Complexity measures and decision tree complexity: A survey.
\newblock {\em Theoretical Computer Science}, 288(1):21--43, 2002.

\bibitem{childs&liu:testing}
A.~Childs and Y-K. Liu.
\newblock Quantum algorithms for testing bipartiteness and expansion of
  bounded-degree graphs.
\newblock Manuscript, Oct 22, 2009.

\bibitem{cdnt:ip}
R.~Cleve, W.~{van} Dam, M.~Nielsen, and A.~Tapp.
\newblock Quantum entanglement and the communication complexity of the inner
  product function.
\newblock In {\em Proceedings of 1st NASA QCQC conference}, volume 1509 of {\em
  Lecture Notes in Computer Science}, pages 61--74. Springer, 1998.
\newblock quant-ph/9708019.

\bibitem{deutsch&jozsa}
D.~Deutsch and R.~Jozsa.
\newblock Rapid solution of problems by quantum computation.
\newblock In {\em Proceedings of the Royal Society of London}, volume A439,
  pages 553--558, 1992.

\bibitem{fgg:nandtreej}
E.~Farhi, J.~Goldstone, and S.~Gutmann.
\newblock A quantum algorithm for the {H}amiltonian {NAND} tree.
\newblock {\em Theory of Computing}, 4(1):169--190, 2008.
\newblock quant-ph/0702144.

\bibitem{FM}
E.~Fischer and A.~Matsliah.
\newblock Testing graph isomorphism.
\newblock {\em SIAM Journal on Computing}, 38(1):207--225, 2008.

\bibitem{fis_sur}
Eldar Fischer.
\newblock The art of uninformed decisions.
\newblock {\em Bulletin of the EATCS}, 75:97, 2001.

\bibitem{fmsp:testhidden}
K.~Friedl, F.~Magniez, M.~Santha, and P.~Sen.
\newblock Quantum testers for hidden group properties.
\newblock {\em Fundamenta Informaticae}, 91(2):325--340, 2009.
\newblock Earlier version in MFCS'03.

\bibitem{ggr}
O.~Goldreich, S.~Goldwasser, and D.~Ron.
\newblock Property testing and its connection to learning and approximation.
\newblock {\em Journal of the ACM}, 45(4):653--750, 1998.

\bibitem{GR}
O.~Goldreich and D.~Ron.
\newblock On testing expansion in bounded-degree graphs.
\newblock {\em Electronic Colloquium on Computational Complexity (ECCC)},
  7(20), 2000.

\bibitem{grover:search}
L.~K. Grover.
\newblock A fast quantum mechanical algorithm for database search.
\newblock In {\em Proceedings of 28th ACM STOC}, pages 212--219, 1996.
\newblock quant-ph/9605043.

\bibitem{holevo}
A.~S. Holevo.
\newblock Bounds for the quantity of information transmitted by a quantum
  communication channel.
\newblock {\em Problemy Peredachi Informatsii}, 9(3):3--11, 1973.
\newblock English translation in {\it Problems of Information Transmission},
  9:177--183, 1973.

\bibitem{inui&legall:testing}
Y.~Inui and F.~Le~Gall.
\newblock Quantum property testing of group solvability.
\newblock In {\em Proceedings of 8th LATIN}, pages 772--783, 2008.

\bibitem{kutin:collision}
S.~Kutin.
\newblock Quantum lower bound for the collision problem with small range.
\newblock {\em Theory of Computing}, 1(1):29--36, 2005.
\newblock quant-ph/0304162.

\bibitem{lachish&newman:periodicity}
O.~Lachish and I.~Newman.
\newblock Testing periodicity.
\newblock {\em Algorithmica}, 2009.
\newblock Earlier version in RANDOM'05.

\bibitem{nielsen&chuang:qc}
M.~A. Nielsen and I.~L. Chuang.
\newblock {\em Quantum Computation and Quantum Information}.
\newblock Cambridge University Press, 2000.

\bibitem{ron_sur}
D.~Ron.
\newblock Property testing: A learning theory perspective.
\newblock {\em Foundations and Trends in Machine Learning}, 1(3):307--402,
  2008.

\bibitem{shor:factoring}
P.~W. Shor.
\newblock Polynomial-time algorithms for prime factorization and discrete
  logarithms on a quantum computer.
\newblock {\em SIAM Journal on Computing}, 26(5):1484--1509, 1997.
\newblock Earlier version in FOCS'94. quant-ph/9508027.

\bibitem{simon:power}
D.~Simon.
\newblock On the power of quantum computation.
\newblock {\em SIAM Journal on Computing}, 26(5):1474--1483, 1997.
\newblock Earlier version in FOCS'94.

\bibitem{Val}
P.~Valiant.
\newblock Testing symmetric properties of distributions.
\newblock In {\em Proceedings of 40th ACM STOC}, pages 383--392, 2008.

\end{thebibliography}

\end{document}